\documentclass{article}

\usepackage[margin=1in]{geometry}
\usepackage{tikz}
\usepackage{booktabs}
\usepackage{amsthm}
\usepackage{amsmath}
\usepackage{amssymb}
\usepackage[hang,flushmargin]{footmisc} 
\usepackage[noline,noend]{algorithm2e}

\usetikzlibrary{positioning,shapes,arrows,matrix}
\pgfdeclarelayer{bg}
\pgfsetlayers{bg,main}

\newtheorem{theorem}{Theorem}
\newtheorem{lemma}{Lemma}
\newtheorem{obs}[theorem]{Observation}
\newtheorem{claim}[theorem]{Claim}
\newtheorem{conjecture}[theorem]{Conjecture}

\newenvironment{definition}[1][Definition]{\begin{trivlist}
		\item[\hskip \labelsep {\bfseries #1}]}{\end{trivlist}}	
\newenvironment{example}[1][Example]{\begin{trivlist}
		\item[\hskip \labelsep {\bfseries #1}]}{\end{trivlist}}
	
\newcommand{\doi}[1]{\url{http://dx.doi.org/#1}}
\newcommand{\R}{\ensuremath{\mathbb{R}}}
\newcommand{\Z}{\ensuremath{\mathbb{Z}}}
\newcommand{\bm}{\mathbf}

\newcommand\bfootnote[1]{\begingroup\renewcommand\thefootnote{}\footnote{#1}\addtocounter{footnote}{-1}\endgroup}

\DeclareMathOperator*{\argmax}{arg\,max}

\title{Dynamics at the Boundary of Game Theory and Distributed Computing}
\author{Aaron D. Jaggard\\U.S. Naval Research Laboratory
\and
Neil Lutz\\Rutgers University
\and
Michael Schapira\\Hebrew University of Jerusalem
\and
Rebecca N. Wright\\Rutgers University}

\begin{document}
	\maketitle
	\bfootnote{This is the authors' version of this work.  The definitive version will be published in \emph{ACM Trans. Econ. Comput.} A preliminary version of some of this work appeared in \emph{Proceedings of the Second Symposium on Innovations in Computer Science ({ICS} 2011)}~\cite{jsw11ics}. Authors' addresses:
	A. D. Jaggard, Formal Methods Section (Code 5543), U.S. Naval Research Laboratory, email: aaron.jaggard@nrl.navy.mil; part of this work was carried out while Jaggard was at DIMACS, Rutgers University and visiting Colgate University;
	N. Lutz, Department of Computer Science, Rutgers University, email:njlutz@rutgers.edu;
	M. Schapira, School of Computer Science and Engineering, Hebrew University of Jerusalem, email: schapiram@huji.ac.il;
	R. N. Wright, Department of Computer Science and DIMACS, Rutgers University, email: rebecca.wright@rutgers.edu.}
	\begin{abstract}
We use ideas from distributed computing and game theory to study dynamic and decentralized environments in which computational nodes, or decision makers, interact strategically and with limited information. In such environments, which arise in many real-world settings, the participants act as both economic and computational entities. We exhibit a general non-convergence result for a broad class of dynamics in asynchronous settings. We consider implications of our result across a wide variety of interesting and timely applications: game dynamics, circuit design, social networks, Internet routing, and congestion control. We also study the computational and communication complexity of testing the convergence of asynchronous dynamics.
Our work opens a new avenue for research at the intersection of distributed computing and game theory.
\end{abstract}
	\section{Introduction}
Dynamic environments where decision makers repeatedly interact arise in a variety of settings, such as Internet protocols, large-scale markets, social networks, and multi-processor computer architectures. Study of these environments lies at the boundary of game theory and distributed computing. The decision makers are both strategic entities with individual economic preferences and computational entities with limited resources, working in a decentralized and uncertain environment. To understand the global behaviors that result from these interactions---the \emph{dynamics} of these systems---we draw on ideas from both disciplines.

The notion of \emph{self stabilization} to a ``legitimate'' state in a distributed system parallels that of convergence to an equilibrium in a game. The foci, however, are very different. In game theory, there is extensive research on dynamics that result from what is perceived as natural strategic decision making (e.g., best- or better-response dynamics, fictitious play, or regret minimization). Even simple heuristics that require little information or computational resources can yield sophisticated behavior, such as the convergence of best-response dynamics to equilibrium points (see \cite{H05} and references therein). These positive results for simple game dynamics are, with few exceptions (see Section~\ref{sec:related}), based on the sometimes implicit and often unrealistic premise of a controlled environment in which actions are synchronous and coordinated. Distributed computing research emphasizes the environmental uncertainty that results from decentralization, but has no notion of ``natural'' rules of behavior. It has long been known that environmental uncertainty---in the form of both asynchrony~\cite{FLP85,Lyn89} and arbitrary initialization~\cite{Dole00}---introduces substantial difficulties for protocol termination in distributed systems. Our work bridges the gap between these two approaches by initiating the study of game dynamics in distributed computing settings. We take the first steps of this research agenda, focusing primarily on systems in which the decision makers, or \emph{computational nodes}, are deterministic and have \emph{bounded recall}, meaning that their behavior is based only on the ``recent history'' of system interaction. Our model is asynchronous in the sense of allowing, at every timestep, an adversarially chosen subset of nodes to be activated.

Our main contribution is a general impossibility result (Theorem~\ref{thm:main}) for asynchronous environments, showing that a large and natural class of bounded-recall dynamics can fail to converge  whenever there are at least two ``equilibrium points'' of the dynamics. We prove this result using a \emph{valency} argument (a now-standard technique in distributed computing theory~\cite{Lyn89,FR03}). We discuss the implications of this result for game dynamics and describe its applications to asynchronous circuit design, social networks, interdomain routing protocols such as BGP, and congestion control in networks. We also explore the impact on convergence guarantees of asynchrony that is bounded, and we present complexity hardness results for checking whether an asynchronous system will always converge: we show that it is PSPACE-hard and requires exponential communication.
	\section{Related Work} \label{sec:related}

Our work relates to many ideas in game theory and in distributed computing.  Here, we discuss game-theoretic work on the dynamics of simple strategies and on asynchrony, distributed computing work on fault tolerance and self stabilization, and other connections between game theory and computer science (for more, see~\cite{Halp03}). We also highlight the application areas we consider.

\vspace{0.04in}\noindent{\bf Algorithmic game theory.} Since our work draws on both game theory and computer science, it may be considered part of the broader research program of algorithmic game theory (AGT), which merges concepts and ideas from those two fields~\cite{NRTV07}. Three main areas of study in AGT have been algorithmic mechanism design, which applies concepts from computer science to economic mechanism design~\cite{NisRon01}; the ``price of anarchy,'' which describes the efficiency of equilibria and draws on approximability research~\cite{KouPap09}; and algorithmic and complexity research on the computation of equilibria~\cite{Papa94}. Analyzing the computational power of learning dynamics in games has been of particular interest (see, e.g.,~\cite{DFPPV10,KLPT11,BaCoMe12,PapPil16}). Our work creates another link between game theory and computer science by drawing on two previously disjoint areas, self-stabilization in distributed computing theory and game dynamics, to explore broader classes of dynamics operating in adversarial distributed environments.

\vspace{0.04in}\noindent{\bf Adaptive heuristics.} Much work in game theory and economics deals with \emph{adaptive heuristics} (see~\cite{H05} and references therein). Generally speaking, this long line of research explores the ``convergence'' of simple and myopic rules of behavior (e.g., best-response/fictitious-play/no-regret dynamics) to an ``equilibrium''. However, with few exceptions (see below), such analysis has so far primarily concentrated on synchronous environments in which steps take place simultaneously or in some other predetermined order. In this work, we explore dynamics of this type in asynchronous environments, which are more realistic for many applications.

\vspace{0.04in}\noindent{\bf Game-theoretic work on asynchronous environments.} Some game-theoretic work on repeated games considers ``asynchronous moves.'' Often, as in \cite{marden2012revisiting}, this asynchrony merely indicates that players are not all activated at each time step, and thus is used to describe environments where only one player is activated at a time (``alternating moves''), or where there is a probability distribution that determines which player(s) are activated at each timestep. Other work does not explore the behavior of dynamics, but has other research goals (e.g., characterizing equilibria, establishing folk theorems); see \cite{LM97,Y04}, among others, and references therein.
To the best of our knowledge, we are the first to study the effects of asynchrony (in the broad distributed computing sense) on the convergence of \emph{game dynamics} to equilibria.

\vspace{0.04in}\noindent{\bf Fault-tolerant computation.} We use ideas and techniques from work in distributed computing on protocol termination in asynchronous computational environments where nodes and communication channels are possibly faulty. Protocol termination in such environments, initially motivated by multi-processor computer architectures, has been extensively studied
in the past three decades~\cite{FLP85,BenOr86,DDS87,BG93,HS99,SZ00}, as nicely surveyed in~\cite{Lyn89,FR03}. Fischer, Lynch and Paterson~\cite{FLP85} showed, in a landmark paper, that a broad class of failure-resilient consensus protocols cannot provably terminate. Intuitively, the risk of protocol non-termination in that work stems from the possibility of failures; a computational node cannot tell whether another node is silent due to a failure or is simply taking a long time to react. Our non-convergence result, by contrast, applies to failure-free environments. In game-theoretic work that incorporated fault tolerance concepts, Abraham et al.~\cite{adgh06podc} studied equilibria that are robust to defection and collusion.

\vspace{0.04in}\noindent{\bf Self stabilization.} The concept of self stabilization is fundamental to distributed computing and dates back to Dijkstra 1974~\cite{Dijkstra74} (see~\cite{Dole00} and references therein). Convergence of dynamics to an ``equilibrium'' in our model can be viewed as the self stabilization of such dynamics (where the ``equilibrium points'' are the legitimate configurations). Our formulation draws ideas from work in distributed computing (e.g., Burns' distributed daemon model~\cite{Burns87}) and in networking research~\cite{GSW02} on self stabilization.

\vspace{0.04in}\noindent{\bf Applications.}
We discuss the implications of our non-convergence result across a wide variety of applications, that have previously been studied: convergence of game dynamics (see, e.g.,~\cite{HarMas03,HarMas06}); asynchronous circuits (see, e.g.,~\cite{DN97}); diffusion of innovations, behaviors, etc., in social networks (see~\cite{Morris00,IKMW07}); interdomain routing~\cite{GSW02,SSZ09}; and congestion control~\cite{gszs10sigmetrics}.
	\section{Asynchronous Dynamic Interaction}\label{sec:asynchronous}

In this section we present our model of asynchronous dynamic interaction. Intuitively, an \emph{interaction system} consists of a collection of computational nodes, each capable of selecting \emph{actions} that are visible to the other nodes. The \emph{state} of the system at any time consists of each node's current action. Each node has a deterministic \emph{reaction function} that maps system histories to actions. At every discrete timestep, each node activated by a schedule simultaneously applies its deterministic reaction function to select a new action, which is immediately visible to all other nodes.

\begin{definition}
	An \emph{interaction system} is characterized by a tuple $(n,A,\bm{f})$:
	\begin{itemize}
		\item The system has $n\in\Z_+$ \emph{computational nodes}, labeled $1,\ldots,n$.
		\item $A=A_1\times\ldots\times A_n$, where each $A_i$ is a finite set called the \emph{action space} of node $i$. $A$ is called the \emph{state space} of the system, and a \emph{state} is an $n$-tuple $\bm{a}=(a_1,\ldots,a_n)\in A$. A \emph{history} of the system is a nonempty finite sequence of states, $H\in A^\ell$, for some $\ell\in\Z_+$. The set of all histories is $A^+=\bigcup_{\ell\in\Z_+}A^\ell$.
		\item $\bm{f}:A^+\to A$ is a function given by $\bm{f}(H)=(f_1(H),\ldots,f_n(H)),$ where $f_i:A^+\to A_i$ is called node $i$'s \emph{reaction function}.
	\end{itemize}
\end{definition}

We now describe the asynchronous dynamics of our model, i.e., the ways that a system's state can evolve due to interactions between nodes. Informally, there is some initial state, and, in each discrete time step $1,2,3,\ldots$, a subset of the nodes are \emph{activated} according to a \emph{schedule}. The nodes that are activated in a given timestep react simultaneously; each applies its reaction function to the current state to choose a new action. This updated action is \emph{immediately observable} to all other nodes.\footnote{This model has ``perfect monitoring.'' While this is clearly unrealistic in some important real-life contexts (e.g., some of the environments considered in Section~\ref{ssec:examples}), this restriction only strengthens our main results, which are impossibility results.} 

\begin{definition}
	Let $S\subseteq[n]$ be a set of nodes. Define the function $\bm{f}_S:A^+\to A$ by
	$\bm{f}_S(H)=(\hat{f}_1(H),\ldots,\hat{f}_n(H))$, where each function $\hat{f}_i:A^{+}\to A_i$ is given by
	\[\hat{f}_i(\bm{a}^0,\ldots,\bm{a}^{\ell})=\left\{\begin{array}{ll}f_i(\bm{a}^0,\ldots,\bm{a}^{\ell})&\text{if }i\in S\\a_i^{\ell}&\text{otherwise}\,.\end{array}\right.\]
	A \emph{schedule} is a function $\sigma:\Z_+\to2^{[n]}$ that maps each $t$ to a (possibly empty) subset of the computational nodes.\footnote{$[n]$ denotes $\{1,\ldots,n\}$, and for any set $S$, $2^S$ is the set of all subsets of $S$.}  If $i\in\sigma(t)$, then we say that node $i$ is \emph{activated} at time $t$.
\end{definition}

Since the reaction functions are deterministic, an initial history and a schedule completely determine the resulting infinite state sequence; we call these state sequences \emph{trajectories}.

\begin{definition}
	Let $H=(\bm{a}^0,\ldots,\bm{a}^{\ell})\in A^+$ be a history, and let $\sigma$ be a schedule.
	The \emph{$(H,\sigma)$-trajectory} of the system is the infinite sequence $\bm{a}^0,\bm{a}^1,\bm{a}^2,\ldots$ extending $H$ such that for every $t>\ell$,
	\begin{equation*}
	\bm{a}^{t}=\bm{f}_{\sigma(t)}(\bm{a}^{0},\ldots,\bm{a}^{t-1})
	\end{equation*}
	The history $(\bm{a}^0,\ldots,\bm{a}^{t-1})$ is the length-$t$ \emph{prefix} of the $(H,\sigma)$-trajectory.
\end{definition}

\subsection{Fairness and Convergence}\label{ssec:fairconv}

Our main theorem is an impossibility result, in which we show that an adversarially chosen initial history and schedule can prevent desirable system behavior. Notice that an arbitrary schedule might never allow some or all nodes to react, or might stop activating them after some time. Hence, we limit this adversarial power (thereby strengthening our impossibility result) by restricting our attention to \emph{fair} schedules, which never permanently stop activating any node.

\begin{definition}\label{def:fair}
	A \emph{fair schedule} is a schedule $\sigma$ that activates each node infinitely many times, i.e., for each $i\in[n]$, the set $\{t\in\Z_+:i\in\sigma(t)\}$ is infinite. A \emph{fair trajectory} is one that is the $(H,\sigma)$-trajectory for some history $H$ and some fair schedule $\sigma$.
\end{definition}

We are especially interested in whether a system's fair trajectories \emph{converge}, eventually remaining at a single state forever.
\begin{definition}\label{def:conv}
	A trajectory $\bm{a}^0,\bm{a}^1,\bm{a}^2,\ldots$ \emph{converges} to a state $\bm{b}$ if there exists some ${T\in\Z_+}$ such that, for all $t>T$, $\bm{a}^t=\bm{b}$. The system is \emph{convergent} if every fair trajectory converges. A state $\bm{b}$ is a \emph{limit state} of the system if some fair trajectory converges to $\bm{b}$.
\end{definition}

Note that it is possible for a trajectory to visit a limit state without converging to that state, meaning that limit states are not necessarily ``stable'' or ``absorbing.'' They may, however, have basins of attraction, in the sense that reaching certain histories might guarantee convergence to a given limit state.

\begin{definition}
	A history $H$ is \emph{committed} to a limit state $\bm{b}$ if, for every fair schedule $\sigma$, the $(H,\sigma)$-trajectory converges to $\bm{b}$. An \emph{uncommitted} history is one that is not committed to any state.
\end{definition}

\subsection{Informational Restrictions on Reaction Functions}
This framework allows for very powerful reaction functions. We now present several possible restrictions on the information they may use. These are illustrated in Fig.~\ref{fig:restrictions}.

Our main theorem concerns systems in which the reaction functions are \emph{self-independent}, meaning that each node ignores its own past and present actions when reacting to the system's state. In discussing self independence, we use the notation
\[A_{-i}=A_1\times\ldots\times A_{i-1}\times A_{i+1}\times\ldots\times A_n\,,\]
the state space of the system when $i$ is ignored. Similarly, for a state $\bm{a}$, $\bm{a}_{-i}\in A_{-i}$ denotes
\[(a_1,\ldots,a_{i-1},a_{i+1},\ldots,a_n)\,,\]
and given a history $H=(\bm{a}^0,\ldots,\bm{a}^{\ell-1})$, we write $H_{-i}$ for $(\bm{a}_{-i}^0,\ldots,\bm{a}_{-i}^{\ell-1})$. Using this notation, we formally define self independence.

\begin{definition}
	A reaction function $f_i$ is \emph{self-independent} if there exists a function $g_i:A_{-i}^+\to A_i$ such that $f_i(H)=g_i(H_{-i})$ for every history $H\in A^+$.
\end{definition}

A reaction function has \emph{bounded recall} if it only depends on recent states.
\begin{definition}\label{def:bddrecall}
	Given $k\in\Z_+$ and a history $H=(\bm{a}^0,\ldots,\bm{a}^{t-1})\in A^t$ with $t\geq k$, the \emph{$k$-history} at $H$ is $H_{|k}:=(\bm{a}^{t-k},\ldots,\bm{a}^{t-1})$, the $k$-tuple of most recent states. A reaction function $f_i$ has \emph{$k$-recall} if it only depends on the $k$-history and the time counter, i.e., there exists a function $g_i:A^{k}\times\Z_+\to A_i$ such that $f_i(H)= g_i(H_{|k},t)$ for every time $t\geq k$ and history $H\in A^t$.
\end{definition}
We sometimes slightly abuse notation by referring to the restricted-domain function $g_i$, rather than $f_i$, as the node's reaction function.

A bounded-recall reaction function is \emph{stationary} if it also ignores the time counter.

\begin{definition}\label{def:stationary}
	We say that a $k$-recall reaction function is \emph{stationary} if the time counter $t$ is of no importance. That is, if there exists a function $g_i:A^{k}\to A_i$ such that $f_i(H)=g_i(H_{|k})$ for every time $t\geq k$ and history $H\in A^t$. A reaction function $f_i$ is \emph{historyless} if $f_i$ is both $1$-recall and stationary. That is, if $f_i$ only depends on the nodes' most recent actions.
\end{definition}

While seemingly very restricted, historyless dynamics capture the prominent and extensively studied best-response dynamics from game theory (as we discuss in Section~\ref{ssec:gamedynamics}). We show in Section~\ref{ssec:examples} that historyless dynamics also encompass a host of other applications of interest, ranging from Internet protocols to the adoption of technologies in social network.

\begin{figure}[h]
	\centering
	\begin{tikzpicture}[scale=1,light node/.style={draw,minimum size=0.75cm,inner sep=0pt},dark node/.style={draw,fill=lightgray,minimum size=0.75cm,inner sep=0pt},loop/.style={looseness=6}]
	\node[dark node] (a1t-3x) {$\ldots$};
	\node[dark node] (a2t-3x) [below=0.1cm of a1t-3x] {$\ldots$};
	\node[light node] (a3t-3x) [below=0.1cm of a2t-3x] {$\ldots$};
	\node[dark node] (a1t-2x) [right=0.1cm of a1t-3x] {\tiny $a_1^{t-2}$};
	\node[dark node] (a2t-2x) [right=0.1cm of a2t-3x] {\tiny $a_2^{t-2}$};
	\node[light node] (a3t-2x) [right=0.1cm of a3t-3x] {\tiny $a_3^{t-2}$};
	\node[dark node] (a1t-1x) [right=0.1cm of a1t-2x] {\tiny $a_1^{t-1}$};
	\node[dark node] (a2t-1x) [right=0.1cm of a2t-2x] {\tiny $a_2^{t-1}$};
	\node[light node] (a3t-1x) [right=0.1cm of a3t-2x] {\tiny $a_3^{t-1}$};
	\node[dark node] [left=0.1cm of a1t-3x] {\tiny $a_1^{0}$};
	\node[dark node] [left=0.1cm of a2t-3x] {\tiny $a_2^{0}$};
	\node[light node] [left=0.1cm of a3t-3x] {\tiny $a_3^{0}$};
	\node[light node] (a1t-3y) [right=4cm of a1t-3x] {$\ldots$};
	\node[light node] (a2t-3y) [below=0.1cm of a1t-3y] {$\ldots$};
	\node[light node] (a3t-3y) [below=0.1cm of a2t-3y] {$\ldots$};
	\node[dark node] (a1t-2y) [right=0.1cm of a1t-3y] {\tiny $a_1^{t-2}$};
	\node[dark node] (a2t-2y) [right=0.1cm of a2t-3y] {\tiny $a_2^{t-2}$};
	\node[dark node] (a3t-2y) [right=0.1cm of a3t-3y] {\tiny $a_3^{t-2}$};
	\node[dark node] (a1t-1y) [right=0.1cm of a1t-2y] {\tiny $a_1^{t-1}$};
	\node[dark node] (a2t-1y) [right=0.1cm of a2t-2y] {\tiny $a_2^{t-1}$};
	\node[dark node] (a3t-1y) [right=0.1cm of a3t-2y] {\tiny $a_3^{t-1}$};
	\node[light node] [left=0.1cm of a1t-3y] {\tiny $a^0_1$};
	\node[light node] [left=0.1cm of a2t-3y] {\tiny $a^0_2$};
	\node[light node] [left=0.1cm of a3t-3y] {\tiny $a^0_3$};
	\node[light node] (a1t-3z) [right=4cm of a1t-3y] {$\ldots$};
	\node[light node] (a2t-3z) [below=0.1cm of a1t-3z] {$\ldots$};
	\node[light node] (a3t-3z) [below=0.1cm of a2t-3z] {$\ldots$};
	\node[light node] (a1t-2z) [right=0.1cm of a1t-3z] {\tiny $a_1^{t-2}$};
	\node[light node] (a2t-2z) [right=0.1cm of a2t-3z] {\tiny $a_2^{t-2}$};
	\node[light node] (a3t-2z) [right=0.1cm of a3t-3z] {\tiny $a_3^{t-2}$};
	\node[dark node] (a1t-1z) [right=0.1cm of a1t-2z] {\tiny $a_1^{t-1}$};
	\node[dark node] (a2t-1z) [right=0.1cm of a2t-2z] {\tiny $a_2^{t-1}$};
	\node[light node] (a3t-1z) [right=0.1cm of a3t-2z] {\tiny $a_3^{t-1}$};
	\node[light node] [left=0.1cm of a1t-3z] {\tiny $a^0_1$};
	\node[light node] [left=0.1cm of a2t-3z] {\tiny $a^0_2$};
	\node[light node] [left=0.1cm of a3t-3z] {\tiny $a^0_3$};	
	\end{tikzpicture}
	\caption{Shading shows the information about past and current actions available to node 3 at time $t$ given different  reaction function restrictions. Left: self-independent. Node 3 can see the entire record of other nodes' past actions, but not its own. The length of this record gives the current timestamp $t$. Center: $2$-recall. Node 3 can see only the two most recent states. Unless the reaction function is stationary, it may also use the value of the current timestamp. Right: self-independent and historyless. Node 3 can only see other nodes' most recent actions and cannot even see the value of the current timestamp.}
	\label{fig:restrictions}
\end{figure}
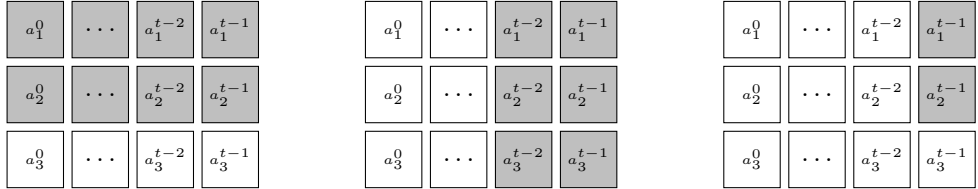
\section{General Non-convergence Result}\label{ssec:bounded}
We now present our main theorem, a general impossibility result for convergence of nodes' actions under bounded-recall dynamics in asynchronous, distributed computational environments.

\begin{theorem}[Main theorem]\label{thm:main}
	In an interaction system where every reaction function is self-independent and has bounded recall, the existence of multiple limit states implies that the system is not convergent.
\end{theorem}

We prove this theorem in Section~\ref{ssec:proof} by using a valency argument. In Section~\ref{ssec:tight}, we show that the hypotheses of Theorem~\ref{thm:main} are necessary. We then discuss in Section~\ref{ssec:consensus} the connections of this work to the famous result of Fischer et al.~\cite{FLP85} on the impossibility of resilient consensus. 

Note that system convergence is closely related to \emph{self stabilization}, which is a guarantee that the system will reach and remaining within a set of \emph{legitimate states}. For a set $L\subseteq A$, we say that a system \emph{self-stabilizes to} $L$ if, for every fair trajectory $\bm{a}^0,\bm{a}^1,\bm{a}^2,\ldots$, there is some $T\in\Z_+$ such that, for every $t>T$, $\bm{a}^t\in L$. In systems satisfying its hypotheses, Theorem~\ref{thm:main} precludes self stabilization to any set containing only committed states.

\subsection{Proof of Main Theorem}\label{ssec:proof}

In proving this theorem, we use the following sequence of lemmas. We first show in Lemma~\ref{lem:kto1} that it is sufficient to consider systems with $1$-recall reaction functions. Then in Lemma~\ref{lem:strongconv}, we argue that such a system can be convergent only if every fair trajectory has no committed prefix. To show the existence of a fair trajectory with no committed prefix, we show that in any such system, uncommitted histories exist (Lemma~\ref{lem:uncom}), and can be extended to a longer uncommitted histories in a way that activates any given node (Lemma~\ref{lem:extend}). This means that committed prefixes can be avoided forever on a trajectory that activates every node infinitely many times, i.e., a fair trajectory.

\begin{lemma}\label{lem:kto1}
	If there exists a convergent interaction system with bounded-recall, self-independent reaction functions and multiple limit states, then there is also a convergent interaction system with $1$-recall, self-independent reaction functions and multiple limit states.
\end{lemma}
\begin{proof}
	Assume that $\Gamma=(n,A,\bm{f})$ is a convergent system with self-independent, $k$-recall reaction functions and multiple limit states, for some $k\in\Z_+$. Consider a $1$-recall system $\Gamma'=(n,A',\bm{f}')$, where $A'=A_1^k\times\ldots\times A_n^k$ and $\bm{f}':A'\times\Z_+\to A'$ is given by
	\[f_i'(
	(a^1_1,\ldots,a^k_1),\ldots,(a^1_n,\ldots,a^k_n),t)=\big[f_i((a^1_1,\ldots,a^1_n),\ldots,(a^k_1,\ldots,a^k_n),kt)\big]^k\,.\]
	Informally, a state in $\Gamma'$ is the transpose of a $k$-history for $\Gamma$. The reaction function $f'_i$ applies $f_i$ to this transpose and repeats the output $k$ times. Notice that $\Gamma'$ has self-independent reaction functions. Furthermore, if $(a_1,\ldots,a_n)$ is a limit state of $\Gamma$, then $((a_1,\ldots,a_1),\ldots,(a_n,\ldots,a_n))$ is a limit state of $\Gamma'$, so $\Gamma'$ also has multiple limit states.
	
	Let $\sigma:\Z_+\to2^{[n]}$ be a fair schedule, and let $H\in (A^k)^\ell$ be a history of $\Gamma'$ for some $\ell\in\Z_+$. Define the schedule $\sigma':\Z_+\to2^{[n]}$ by
	\[\sigma'(t)=\left\{\begin{array}{ll}\sigma(t/k)&t/k\in\Z_+\\\emptyset&\text{otherwise}\,.\end{array}\right.\]
	Notice that $\sigma'$ is also fair. Let $H'\in A^{k\ell}$ be the history for $\Gamma$ formed by concatenating the $k$-tuples in $H$. It is easy to see that the $(H,\sigma)$-trajectory of $\Gamma'$ converges if and only if the $(H',\sigma')$-trajectory of $\Gamma$ converges. Since we assumed that $\Gamma$ is convergent, it follows that $\Gamma'$ is also.
\end{proof}

\begin{lemma}\label{lem:strongconv}
	Let $\Gamma$ be a convergent system with self-independent, $1$-recall reaction functions. Then every fair trajectory in $\Gamma$ has a committed finite prefix.
\end{lemma}
\begin{proof}
	Assume there exist some history $H$ and fair schedule $\sigma$ for $\Gamma$ such that the $(H,\sigma)$-trajectory converges to a state $\bm{a}=(a_1,\ldots,a_n)$ but has no committed finite prefix. We will construct a fair schedule $\sigma^\prime$ such that the $(H,\sigma^\prime)$-trajectory does not converge, giving a contradiction.
	
	Let $\bm{u}^0,\bm{u}^1,\bm{u}^2,\ldots$ be the $(H,\sigma)$-trajectory. Then there is some $t_0\in\Z_+$ such that $\bm{u}^t=\bm{a}$ for all $t\geq t_0$. The fairness of $\sigma$ implies that there is some $t_1> t_0$ such that every node is activated by $\sigma$ between $t_0$ and $t_1$, i.e., $\bigcup_{t_0<t<t_1}\sigma(t)=[n]$. By assumption, $(\bm{u}^0,\ldots,\bm{u}^{t_1})$ is not committed to $\bm{a}$, which means there is some time $t_2\geq t_1$ and node $i\in[n]$ such that $f_i(\bm{a},t_2)\neq a_i$.
	The fairness of $\sigma$ also implies that there is some $t_3> t_2$ such that $i\in\sigma(t_3)$. Since $t_3\geq t_0$, we must have $f_i(\bm{a},t_3)=a_i$. By self-independence, then, $f_i(\bm{a}',t_3)=a_i$ for all $\bm{a}'$ such that $\bm{a}'_{-i}=\bm{a}_{-i}$.
	
	We use these facts to iteratively build our fair schedule $\sigma^\prime$. In the $(H,\sigma^\prime)$-trajectory $\bm{v}^0,\bm{v}^1,\bm{v}^2,\ldots$, the system will repeatedly enter and exit the state $\bm{a}$.
	First, let $\sigma^\prime(t)=\sigma(t)$ for all $1\leq t\leq t_0$, so that $\bm{v}^{t_0}=\bm{a}$. Define $\sigma'(t)$ on values $t_0<t\leq t_2$ as follows.
	\[\sigma'(t)=\left\{\begin{array}{ll}\sigma(t)&t_0<t<t_2\\
	\{i\}&t_2\leq t\leq t_3\,.\end{array}\right.\]
	By our choices of $t_0$, $t_2$, and $t_3$, this partial schedule activates every node and induces a segment \[(\bm{v}^{t_0+1},\ldots,\bm{v}^{t_3})\]
	of the $(H,\sigma^\prime)$-trajectory such that $\bm{v}^{t}=\bm{a}$ whenever $t_0<t<t_2$ or $t=t_3$, but $\bm{v}^{t_2}\neq \bm{a}$.
	
	Now set $t_0=t_3$, select new $t_1$, $t_2$, $i$, and $t_3$ relative to this $t_0$, and iterate this process to define $\sigma'(t)$ for all values $t\in\Z_+$. Notice that $\sigma'$ is fair and that the $(H,\sigma^\prime)$-trajectory $\bm{v}^0,\bm{v}^1,\bm{v}^2,\ldots$ does not converge, which contradicts the assumption that the system is convergent. Therefore every fair trajectory in the system must have a committed finite prefix.
\end{proof}

We will use the following consequence of self independence in the course of proving Lemmas~\ref{lem:uncom} and~\ref{lem:extend}.

\begin{obs}\label{obs:commit}
	Let $H'=(\bm{a}^0,\ldots,\bm{a}^{\ell})$ and $H'=(\bm{b}^0,\ldots,\bm{b}^{\ell})$ be committed histories in an system with self-independent, 1-recall reaction functions. If $\bm{a}_{-i}^{\ell}=\bm{b}_{-i}^{\ell}$ for some $i\in[n]$, then $H$ and $H'$ are committed to the same limit state.
\end{obs}
\begin{proof}
	Let $H=(\bm{a}^0,\ldots,\bm{a}^{\ell})$ and $H'=(\bm{b}^0,\ldots,\bm{b}^{\ell})$ be committed histories such that, for some $i\in[n]$, $\bm{a}^\ell_{-i}=\bm{b}^\ell_{-i}$, as in Fig.~\ref{fig:commit}. Let $\sigma$ be any fair schedule such that ${\sigma(\ell+1)=\{i\}}$, and consider the $(H,\sigma)$- and $(H^\prime,\sigma)$-trajectories. When node $i$ is activated, it will choose the same action regardless of whether the history is $H$ or $H'$, by self independence. As the reaction functions have $1$-recall, this means that both these trajectories are identical after time $\ell+1$. Thus, since $H$ and $H'$ are both committed, they must be committed to the same limit state.
\end{proof}
\begin{figure}[ht]
	\centering
	\begin{tikzpicture}[scale=1,main node/.style={ellipse,draw,thick,minimum size=0.75cm,inner sep=0pt}]
	\node[main node] (a) {$\bm{a}^\ell=(a^\ell_1,\ldots,a^\ell_i,\ldots,a^\ell_n)$};
	\node[main node] (aprime) [below of=a]{$\bm{b}^\ell=(a^\ell_1,\ldots,b^\ell_i,\ldots,a^\ell_n)$};
	\node[main node] (c) [below right = 0cm and 2.5cm of a]{};
	\path[draw,thick,->]
	(a) edge node [above] {$\{i\}$} (c)
	(aprime) edge node [below] {$\{i\}$} (c)
	;
	\end{tikzpicture}
	\caption{Activating $\{i\}$ from $H=(\bm{a}^0,\ldots,\bm{a}^{\ell})$ or $H'=(\bm{b}^0,\ldots,\bm{b}^{\ell})$ will have the same outcome.}
	\label{fig:commit}
\end{figure}
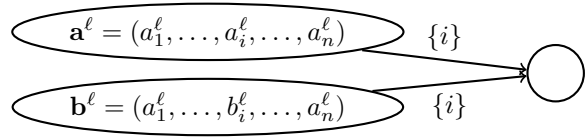

\begin{lemma}\label{lem:uncom}
	Every interaction system with 1-recall, self-independent reaction functions and more than one limit state has at least one uncommitted history.
\end{lemma}
\begin{proof}
	Suppose that every history of length one is committed, and consider two such histories $(\bm{a})=((a_1,\ldots,a_n))$ and $(\bm{b})=((b_1,\ldots,b_n))$. Observation~\ref{obs:commit} implies that, for all $1\leq i<n$, the histories $((a_1,\ldots,a_{i-1},b_{i},\ldots,b_n))$ and $((a_1,\ldots,a_i,b_{i+1},\ldots,b_n))$ are committed to the same limit state, and therefore that $\bm{a}$ and $\bm{b}$ are committed to the same limit state. Thus, all histories of length one must be committed to the same limit state, and it follows that all histories must be committed to the same limit state. This contradicts the system having more than one limit state.
\end{proof}
\begin{lemma}\label{lem:extend}
	Let $(n,A,\bm{f})$ be an interaction system with self-independent, 1-recall reaction functions and more than one limit state, let $H= (\bm{a}^0, \ldots, \bm{a}^{\ell-1})\in A^\ell$ be an uncommitted history, for some $\ell\in\Z_+$, and let $i\in[n]$ be a node. Then there exist some $t\geq\ell$ and schedule $\sigma$ such that $i\in\sigma(t)$ and the length-$(t+1)$ prefix of the $(H,\sigma)$-trajectory is uncommitted.
\end{lemma}
\begin{proof}
	Assume for contradiction that no such $t$ and $\sigma$ exist. Consider all histories that result from activating a set containing $i$ at history $H$. By assumption, each of these histories is committed. Notice that for all $S\subseteq[n]$ and $j\in[n]$, the states $\bm{f}_S(H)$ and $\bm{f}_{S\cup\{j\}}(H)$ can only differ at coordinate $j$. Hence, we can iteratively apply Observation~\ref{obs:commit}, much as in the proof of Lemma~\ref{lem:uncom}, to see that all these histories must be committed to the same limit state, which we call $\bm{b}$.
	
	Let $\sigma$ be any fair schedule, and let $\bm{a}^0,\bm{a}^1,\ldots$ be the $(H,\sigma)$-trajectory. For each $t\in\Z_+$, let $H^t=(\bm{a}^0,\ldots,\bm{a}^t)$, and notice that $H^{\ell-1}=H$. For each $t\geq\ell$, 
	let $\bm{v}^t=\bm{f}_{\sigma(t)\cup\{i\}}(H^{t-1})$, and let $I^t=(H^{t-1},\bm{v}^{t})$.
	Since $i\in\sigma(t)\cup\{i\}$, our assumption implies that each history $I^t$ is committed. Let $\bm{w}^t=\bm{f}_{\{i\}}(H^{t-1})$, and note that by self independence, $\bm{f}_{\{i\}}(I^{t-1})=\bm{w}^t$ also, as illustrated in Fig.~\ref{fig:traj}. Let $J^t=(I^{t-1},\bm{w}^{t})$.
	
	We now show by induction on $t$ that, for every $t\geq\ell$, the history $I^t$ is committed to $\bm{b}$. This holds for $t=\ell$ by our definition of $\bm{b}$. Fix $t>\ell$, and suppose that $I^{t-1}$ is committed to $\bm{b}$. Then $J^{t}$ is also committed to $\bm{b}$. Consider all histories that result from activating a set containing $i$ at history $H^{t-1}$. As before, our assumption implies that all these histories are committed, and iterative application of Observation~\ref{obs:commit} shows that they are all committed to the same limit state. In particular, $I^t$ must be committed to the same limit state as $J^{t}$, namely $\bm{b}$.
	
	Since $\sigma$ is a fair schedule, there is some time $t$ for which $i\in\sigma(t)$. For this $t$, we have $H^t=I^t$, so $H^t$ is committed to $\bm{b}$. Thus for every fair schedule $\sigma$, the $(H,\sigma)$-trajectory converges to $\bm{b}$, contradicting the assumption that $H$ is uncommitted. We conclude that our assumption was false and that the lemma holds.
\end{proof}

\begin{figure}[h]
	\centering
	\begin{tikzpicture}[scale=1,main node/.style={circle,draw,thick,minimum size=1cm,inner sep=0pt}]
	\node[main node] (a) {$H$};
	\node[main node] (u1) [right = 3cm of a] {$H^{\ell}$};
	\node[main node] (w1) [below = 0.75cm of u1] {$J^{\ell+1}$};
	\node[main node] (v1) [left = 1.5cm of w1] {$I^{\ell}$};
	\node[main node] (u2) [right = 3cm of u1] {$H^{\ell+1}$};
	\node[main node] (w2) [below = 0.75cm of u2] {$J^{\ell+2}$};
	\node[main node] (v2) [left = 1.5cm of w2] {$I^{\ell+1}$};
	\node (rightdots) [below right = 0.5cm and 0cm of u2] {$\displaystyle{\ldots}$};
	\path[draw,thick,->]
	(a) edge node [above] {$\sigma(\ell)$} (u1)
	(a) edge node [right] {$\sigma(\ell)\cup\{i\}$} (v1)
	(u1) edge node [left] {$\{i\}$} (w1)
	(v1) edge node [below] {$\{i\}$} (w1)
	(u1) edge node [above] {$\sigma(\ell+1)$} (u2)
	(u1) edge node [right] {$\sigma(\ell+1)\cup\{i\}$} (v2)
	(u2) edge node [left] {$\{i\}$} (w2)
	(v2) edge node [below] {$\{i\}$} (w2)
	;
	\end{tikzpicture}
	\caption{All histories in the bottom row are committed to the same limit state $\bm{b}$.}
	\label{fig:traj}
\end{figure}
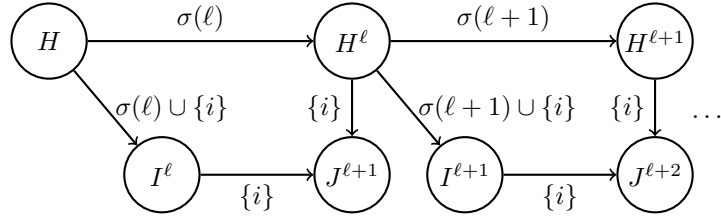

\begin{proof}[Proof of Theorem~\ref{thm:main}]

It follows from Lemmas~\ref{lem:uncom} and~\ref{lem:extend} that, in every system with $1$-recall, self-independent reaction functions and multiple limit states, it is be possible to activate each node infinitely many times without ever reaching a committed history. This means that every such system has a fair trajectory with no committed prefix. By Lemma~\ref{lem:strongconv}, this implies that no such system may be convergent. The theorem follows immediately by  Lemma~\ref{lem:kto1}.\footnote{Although our primary focus is on discrete state spaces, we note that in continuous metric spaces, the standard notion of convergence only requires indefinite approach; the limit point might never be reached. Accordingly, given a metric $d$ on an infinite state space $A$, one could modify Definition~\ref{def:conv} to say that a trajectory $\mathbf{a}^0,\mathbf{a}^1,\mathbf{a}^2,\ldots$ converges to a state $\mathbf{b}$ if for every $\varepsilon>0$ there exists some $T\in\Z_+$ such that, for all $t>T$, $d(\mathbf{a}^t,\mathbf{b})<\varepsilon$. If we require every limit state to have a committed neighborhood, then our proof of Theorem~\ref{thm:main} still holds in this setting.}
\end{proof}

\subsection{Tightness of Main Theorem}\label{ssec:tight}
The following two examples demonstrate that the statement of Theorem~\ref{thm:main} does not hold if either the self-independence restriction or the bounded-recall restriction is removed.

\begin{example}\label{ex:iindependent} {\it The self-independence restriction cannot be removed.}
	
	\noindent Consider a system with one node, with action space $\{\alpha,\beta\}$. When activated, the node always re-selects its own current action. Observe that the system is convergent despite having two limit states.
\end{example}

\begin{example}\label{ex:not-consensus} {\it The bounded-recall restriction cannot be removed.}
	
	\noindent Consider a system with two nodes, $1$ and $2$, each with the action space $\{\alpha,\beta\}$. The self-independent reaction functions of the nodes are as follows: node 2 always chooses node 1's action; node 1 will choose $\beta$ if node 2's action changed from $\alpha$ to $\beta$ in the past, and $\alpha$ otherwise. Observe that node 1's reaction function has unbounded recall: it depends on the entire history of interaction. We make the observations that the system is convergent and has two limit states. Observe that if node 1 chooses $\beta$ at some point in time due to the fact that node 2's action changed from $\alpha$ to $\beta$, then it will continue to do so thereafter; if, on the other hand, $1$ never does so, then from some point in time onwards, node 1's action is constantly $\alpha$. In both cases, node 2 will have the same action as node 1 eventually, and thus convergence to one of the two limit states, $(\alpha,\alpha)$ and $(\beta,\beta)$, is guaranteed. Hence, two limit states exist and the system is convergent nonetheless. Notice also that node 1's reaction functions requires only two states, so the bounded-recall restriction cannot be replaced by a memory restriction.
\end{example}
	\subsection{Connection to Consensus Protocols}\label{ssec:consensus}

We now discuss the relationship of our main result (Theorem~\ref{thm:main}) to the seminal result of Fischer et al. on the impossibility of \emph{fault-resilient consensus protocols}. The consensus problem is fundamental to distributed computing research. We give a brief description of it here, and we refer the reader to~\cite{FLP85} for a detailed explanation of the model.

Fischer et al. studied an environment in which a group of \emph{processes}, each with an initial value in $\{0,1\}$, communicate with each other via \emph{messages}. The objective is for all \emph{non-faulty} processes to eventually agree on some \emph{consensus} value $x\in\{0,1\}$, where $x$ must match the initial value of some process. Fischer et al. established that no consensus protocol is resilient to even a single failure. Their proof of this breakthrough non-termination result introduced the idea of a \emph{valency} argument. They showed that there exists some initial configuration that is \emph{bivalent}, meaning that the resulting consensus could be either $0$ and $1$ (the outcome depends on the asynchronous schedule of message transmission), and that this bivalence can be maintained. Our proof of Theorem~\ref{thm:main} also uses a valency argument, where uncommitted histories play the role of bivalent configurations.

Intuitively, the risk of protocol non-termination in the environment studied by Fischer et al. stems from the possibility of failures; a computational node cannot tell whether another node is silent due to a failure or is simply taking a long time to react. Our non-convergence result concerns environments in which nodes/communication channels do not fail. Thus, each node is guaranteed that all other nodes will eventually react. Observe that in such an environment reaching a consensus is easy; one pre-specified node $i$ (the ``dictator'') waits until it learns all other nodes' inputs (this is guaranteed to happen as failures are impossible) and then selects a value $v_i$ and informs all other nodes; then, all other nodes select $v_i$. By contrast, the possibility of non-convergence shown in Theorem~\ref{thm:main} stems from limitations on nodes' behaviors. Hence, there is no immediate translation from the result of Fischer et al. to ours (and vice versa).

	\section{Applications: Games, Circuits, Social Networks, and Routing}\label{ssec:examples}

We present implications of our impossibility result, Theorem~\ref{thm:main}, for several well-studied environments: game dynamics, circuit design, social networks, and Internet protocols. For most of these applications, the reaction functions are historyless. Recalling Definition~\ref{def:stationary}, this means that they react only to the current state of the system.

\subsection{Asynchronous Game Dynamics}\label{ssec:gamedynamics}
Traditionally, work in game theory on game dynamics (e.g., best-response dynamics) relies on the explicit or implicit premise that players' behavior is somehow synchronized (in some contexts play is sequential, while in others it is simultaneous). Here, we consider the realistic scenario that there is no computational center than can synchronize players' selection of strategies. We describe these dynamics in the setting of this work and exhibit an impossibility result for best-response, and more general, dynamics.

A \emph{game} is characterized by a triple $(n,S,\bm{u})$. There are $n$ \emph{players}, $1,\ldots,n$. Each player $i$ has a \emph{strategy set} $S_i$. $S=S_1\times\ldots\times S_n$ is the space of \emph{strategy profiles} $\bm{s}=(s_1,\ldots,s_n)$. Each player $i$ has a \emph{utility function} $u_i:S\to\R$, where $\bm{u}=(u_1\ldots u_n)$. Intuitively, player $i$ ``prefers'' states for which $u_i$ is higher. Informally, a player is \emph{best responding} when it has no incentive to unilaterally change its strategy.
\begin{definition}
	In a game $U=(n,S,\bm{u})$, player $i$ is \emph{best responding} at $\bm{s}\in S$ if ${u_i(\bm{s})\geq u_i(\bm{s}^\prime)}$ for every $\bm{s}^\prime\in S$ such that $s_{-i}=s^\prime_{-i}$. We write $s_i\in BR^U_i(\bm{s})$.
	A strategy profile $\bm{s}\in S$ is a \emph{pure Nash equilibrium} (\emph{PNE}) if every player is best responding at $\bm{s}$.
\end{definition}

There is a natural relationship between games and the interaction systems described in Section~\ref{sec:asynchronous}. A player with a strategy set corresponds directly to a node with an action space, and a strategy profile may be viewed as a state. These correspondences are so immediate that we often use these terms interchangeably.

Consider the case of \emph{best-response dynamics} for a game in which best responses are unique (a \emph{generic} game): starting from some arbitrary strategy profile, each player chooses its unique best response to other players' strategies when activated. Convergence to pure Nash equilibria under best-response dynamics is the subject of extensive research in game theory and economics, and both positive~\cite{R73,MS96} and negative~\cite{HarMas03,HarMas06} results are known. If we view each player $i$ in a game $(n,S,\bm{u})$ as a node in an interaction system, then under best-response dynamics its utility function $u_i$ induces a self-independent historyless reaction function $f_i:S_{-i}\to S_i$, as long as best responses are unique. Formally,
$$f_i(\bm{a}_{-i})=\argmax_{\alpha\in S_i}u_i(a_1,\ldots,\alpha,\ldots,a_n)\,.$$

Conversely, any system with historyless and self-independent reaction functions can be described as following best-response dynamics for a game with unique best responses. Given  reaction functions $f_1,\ldots,f_n$, consider the game where each player $i$'s utility function is given by
$$
u_i(\bm{a})=
\left\{
\begin{array}{ll}
1&\;\mbox{if }f_i(\bm{a})=a_i\\
0&\;\mbox{otherwise}\,.
\end{array}
\right.	
$$
Best-response dynamics on this game replicate the dynamics induced by those reaction functions. Thus historyless and self-independent dynamics are exactly equivalent to best-response dynamics. Since pure Nash equilibria are fixed points of these dynamics, the historyless case of Theorem~\ref{thm:main} may be restated in the following form.

\begin{theorem}\label{thm:bestresponse}
	If there are two or more pure Nash equilibria in a game with unique best responses, then asynchronous best-response dynamics can potentially oscillate indefinitely.
\end{theorem}

In fact, best-response dynamics are just one way to derive reaction functions from utility functions, i.e., to translate preferences into behaviors. In general, a \emph{game dynamics protocol} is a mapping from games to systems that makes this translation. Given a game $(n,S,\bm{u})$ as input, the protocol selects reaction functions $\bm{f}=(f_1,\ldots,f_n)$, and returns an interaction system $(n,S,\bm{f})$. The above non-convergence result holds for a large class of these protocols. In particular, it holds for bounded-recall and self-independent game dynamics, whenever pure Nash equilibria are limit states. When cast into game-theoretic terminology, Theorem~\ref{thm:main} says that if players' choices of strategies are not synchronized, then the existence of two (or more) pure Nash equilibria implies that this broad class of game dynamics are not guaranteed to reach a pure Nash equilibrium. This result should be contrasted with positive results for such dynamics in the traditional synchronous game-theoretic environments. In particular, this result applies to best-response dynamics with bounded recall and consistent tie-breaking rules (studied by Zapechelnyuk~\cite{Z08}).

\begin{theorem}
	If there are two or more pure Nash equilibria in a game with unique best responses, then all bounded-recall self-independent dynamics for which those equilibria are fixed points can fail to converge in asynchronous environments.
\end{theorem}

\subsection{Asynchronous Circuits}
The implications of asynchrony for circuit design have been extensively studied in computer architecture research~\cite{DN97}. By regarding each logic gate as a node executing an inherently historyless and self-independent reaction function, we show that an impossibility result for stabilization of asynchronous circuits follows from Theorem~\ref{thm:main}.

In this setting there is a Boolean circuit, represented as a directed graph $G$, in which the vertices represent the circuit's inputs and logic gates, and the edges represent the circuit's connections. The activation of the logic gates is asynchronous. That is, the gates' outputs are initialized in some arbitrary way, and then the update of each gate's output, given its inputs, is uncoordinated and unsynchronized. A \emph{stable Boolean assignment} in this framework is an assignment of Boolean values to the circuit inputs and the logic gates that is consistent with each gate's truth table. We say that a Boolean circuit is \emph{inherently stable} if it is guaranteed to converge to a stable Boolean assignment regardless of the initial Boolean assignment.

To show how Theorem~\ref{thm:main} applies to this setting, we model an asynchronous circuit with a fixed input as a historyless interaction system with self-independent reaction functions. Every node in the system has action space $\{0,1\}$. There is a node for each input vertex, which has a constant (and thus self-independent) reaction function. For each logic gate, the system includes a node whose reaction function implements the logic gate on the actions of the nodes corresponding to its inputs. If any gate takes its own output directly as an input, we model this using an additional identity gate; this means that all reaction functions are self-independent. Since every stable Boolean assignment corresponds to a limit state of this system, this instability result follows from Theorem~\ref{thm:main}:

\begin{theorem}
	If two or more stable Boolean assignments exist for an asynchronous Boolean circuit with a given input, then that asynchronous circuit is not inherently stable on that input.
\end{theorem}

\subsection{Diffusion of Technologies in Social Networks}
Understanding the ways in which innovations, ideas, technologies, and practices disseminate through social networks is fundamental to the social sciences. We consider the classic economic setting~\cite{Morris00} (which has lately also been approached by computer scientists~\cite{IKMW07}) where each decision maker has two technologies $\{X,Y\}$ to choose from and wishes to have the same technology as the majority of its ``friends'' (neighboring nodes in the social network). We exhibit a general asynchronous instability result for this environment.

In this setting there is a social network of users, represented by a connected graph in which users are the vertices and edges correspond to friendship relationships. There are two competing technologies, $X$ and $Y$. Each user will repeatedly reassess his choice of technology, at timesteps separated by arbitrary finite intervals. When this happens, the user will select $X$ if at least half of his friends are using $X$ and otherwise select $Y$. A ``stable global state'' is a fixed point of these choice functions, meaning that no user will ever again switch technologies. Observe that if every user has chosen $X$ or every user has chosen $Y$, then the system is in a stable global state.

The dynamics of this diffusion can be described as asynchronous best-response dynamics for the game in which each player's utility is $1$ if his choice of technology is consistent with the majority (with ties broken in favor of $X$) and $0$ otherwise. This game has unique best responses, and the strategy profiles $(X,\ldots,X)$ and $(Y,\ldots,Y)$ are both pure Nash equilibria for this game. Thus Theorem~\ref{thm:bestresponse} implies the following result.
\begin{theorem}
	In every social network with at least one edge, the diffusion of technologies can potentially oscillate indefinitely.
\end{theorem}

\subsection{Interdomain Routing}
\emph{Interdomain routing} is the task of establishing routes between the smaller networks, or \emph{autonomous systems} (ASes), that make up the Internet. It is handled by the \emph{Border Gateway Protocol} (BGP). We abstract a recent result of Sami et al.~\cite{SSZ09} concerning BGP non-convergence and show that this result extends to several BGP-based \emph{multipath routing} protocols that have been proposed in the past few years.

In the standard model for analyzing BGP dynamics~\cite{GSW02}, there is a network of \emph{source} ASes that wish to send traffic to a unique \emph{destination} AS $d$. Each AS $i$ has a \emph{ranking function} $<_i$ that specifies $i$'s strict preferences over all simple (loop-free) routes leading from $i$ to $d$. Each AS also has an \emph{export policy} that specifies which routes it is willing to make available to each neighboring AS. Under BGP, each AS constantly selects the ``best'' route that is available to it (see~\cite{GSW02} for more details). \emph{BGP safety}, i.e., guaranteed convergence to a stable routing outcome, is a fundamental desideratum that has been the subject of extensive work in both the networking and the standards communities. We now cast interdomain routing into the terminology of Section~\ref{sec:asynchronous} to obtain a non-termination results for BGP as a corollary of Theorem~\ref{thm:main}.

Each AS acts as a computational node. The action space of each node $i$ is the set of all simple routes from $i$ to the destination $d$, together with the empty route $\emptyset$. For every state $(R_1,\ldots,R_n)$ of the system, $i$'s reaction function $f_i$ considers the set of routes $S=\{(i,j)R_j : j\mbox{ is }i\mbox{'s neighbor and }R_j\mbox{ is exportable to }i\}$.
If $S$ is empty, $f_i$ returns $\emptyset$. Otherwise, $f_i$ selects the route in $S$ that is optimal with respect to $<_i$. Observe that this reaction function is deterministic, self-independent, and historyless, and that a stable routing tree is a limit state of this system. Thus Theorem~\ref{thm:main} implies the following result of Sami et al.

\begin{theorem}[Sami et al.~\cite{SSZ09}]\label{thm:sami}
	If there are multiple stable routing trees in a network, then BGP is not safe on that network.
\end{theorem}

Importantly, the asynchronous model of Section~\ref{sec:asynchronous} is significantly \emph{more restrictive} than that of Sami et al., so the result implied by Theorem~\ref{thm:main} is stronger. Theorem~\ref{thm:main} implies an even more general non-safety result for routing protocols that depend on ASes that act self-independently and with bounded recall. In particular, this includes recent proposals for BGP-based \emph{multi-path routing} protocols that allow each AS to send traffic along multiple routes, e.g., R-BGP~\cite{KKKM07} and Neighbor-Specific BGP (NS-BGP)~\cite{YSR09}.

\subsection{Congestion Control}
We now consider the fundamental task of congestion control in communication networks, which is achieved through a combination of mechanisms on \emph{end-hosts} (e.g., TCP) and on \emph{switches}/\emph{routers} (e.g., RED and WFQ). We briefly describe the model of congestion control studied by Godfrey et al.~\cite{gszs10sigmetrics}.

There is a network of routers, represented by a directed graph $G$, in which vertices represent routers, and edges represent communication links. Each edge $e$ has capacity $c_e$. There are $n$ \emph{source-target pairs} of vertices $(s_i,t_i)$, termed ``\emph{connections}'', that represent communicating pairs of end-hosts. Each source-target pair $(s_i,t_i)$ is connected via some \emph{fixed} route, $R_i$. Each source $s_i$ transmits at a \emph{constant} rate $\gamma_i>0$.\footnote{This is modeled via the addition of an edge $e=(u,s_i)$ to $G$, such that $c_e=\gamma_i$, and $u$ has no incoming edges.} For each of a router's outgoing edges, the router has a \emph{queue management}, or \emph{queuing}, policy, that dictates how the edge's capacity should be allocated between the connections whose routes traverse that edge. The network is asynchronous, so routers' queuing decisions can be made simultaneously. An \emph{equilibrium} of the network is a global configuration of edges' capacity allocation such that the incoming and outgoing flows on each edge are consistent with the queuing policy for that edge. Godfrey et al. show that, while one might expect flow to be received at a constant rate whenever it is transmitted at a constant rate, this is not necessarily the case. Indeed, Godfrey et al. present examples in which connections' throughputs can potentially fluctuate ad infinitum, never converging an equilibrium.

We model such a network as a historyless interaction system to show that \emph{every} network with multiple equilibria can oscillate indefinitely. The computational nodes of the system are the edges. The action space of each edge $e$ intuitively consists of all possible ways to divide traffic going through $e$ between the connections whose routes traverse $e$. More formally, for every edge $e$, let $N(e)$ be the number connections whose paths go through $e$. Then $e$'s action space is $A_e=\left\{\left(x_1,\ldots,x_{N(e)}\right)\;:\;\text{each}\ x_i\geq0\ \text{and}\ \Sigma_i x_i\leq c_e\right\}$.
We assume that the $x_i$ have bounded precision, meaning that the state space is finite.

Edge $e$'s reaction function $f_e$ models the queuing policy according to which $e$'s capacity is shared: for every $N(e)$-tuple of nonnegative incoming flows $(w_1,\ldots, w_{N(e)})$, $f_e$ outputs an action $(x_1,\ldots,x_{N(e)})\in A_e$ such that for every $i\in [N(e)]$, $x_i \leq w_i$---a connection's flow leaving the edge cannot exceed its flow entering the edge. These reaction functions are historyless and self-independent, and an equilibrium of the network is a limit state of this system. Using Theorem~\ref{thm:main}, then, we can obtain the following impossibility result.

\begin{theorem}
	If there are multiple capacity-allocation equilibria in the network, then dynamics of congestion control can potentially oscillate indefinitely.
\end{theorem}

	\section{Complexity of Asynchronous Dynamics}\label{ssec:complexity}

We now turn to the communication complexity and computational complexity of determining whether a system is convergent. We present hardness results in both models of computation even for the case of historyless interaction. Our computational complexity result shows that even if nodes' reaction functions can be succinctly represented, determining whether the system is convergent is PSPACE-complete. Alongside its computational implications, this intractability result implies that (unless PSPACE $\subseteq$ NP) we cannot hope to have short, efficiently verifiable certificates that guarantee a system's convergence.

\subsection{Communication Complexity}

The following result shows that, in general, determining whether a system is convergent cannot be done efficiently.
\begin{theorem}\label{thm:com-cxity}
	Determining if a system with $n$ nodes, each with $2$ actions, is convergent requires communicating
	$\Omega(2^{n})$ bits. This holds even if all nodes have historyless and self-independent
	reaction functions.
\end{theorem}
	\begin{proof}
		To prove our result we present a reduction from the \emph{2-party set disjointness problem}, a well-known problem in communication complexity theory: There are two parties, Alice and Bob. Each party holds a subset of $[q]$; Alice holds the subset $A$ and Bob holds the subset $B$. The objective is to determine whether $A\cap B=\emptyset$. This problem instance is denoted $\textsc{Disj}^q(A,B)$. The following is well known.
		\begin{theorem}
			Determining whether $A,B\subseteq[q]$ are disjoint requires (in the worst case) the communication of $\Omega(q)$ bits. This lower bound applies to randomized protocols with bounded $2$-sided error and also to nondeterministic protocols.
		\end{theorem}
	
	We now present a reduction from 2-party set disjointness to the question of determining whether a system with historyless and self-independent reaction functions is convergent. Given an instance $\textsc{Disj}^q(A,B)$, we construct a system with $n$ nodes, each with two actions, as follows. (The relation between the parameters $q$ and $n$ is to be specified later.) Let the action space of each node be $\{0,1\}$. We now define the reaction functions of the nodes. Consider the possible action profiles of nodes $3,\ldots,n$, i.e., the set $\{0,1\}^{n-2}$. Observe that this set of actions profiles is the $(n-2)$-hypercube $Q_{n-2}$, and thus can be visualized as the graph whose vertices are indexed by the binary $(n-2)$-tuples and such that two vertices are adjacent if and only if they differ in exactly one coordinate. The reaction functions are based on following a \emph{snake} in this hypercube.
	
	\begin{definition}\label{def:snake}
		A \emph{snake} in a hypercube $Q_n$ is a simple cycle $S=(v_0,\ldots,v_k)$ that is \emph{chordless}, i.e., for each $v_i,v_j$ on $S$, if $v_i$ and $v_j$ are neighbors in $Q_n$, then $v_j\in \{v_{i-1},v_{i+1}\}$.
	\end{definition}

	Let $S$ be a maximal snake in $Q_{n-2}$, and let $q=|S|$. We now show our reduction from $\textsc{Disj}^q$. We identify each element $j\in [q]$ with a unique vertex $v^j\in S$. Without loss of generality, we assume that $0^{n-2}$ is on $S$. For ease of exposition we also assume that $1^{n-2}$ is not on $S$. (Getting rid of this assumption is easy.) Orient the edges in $S$ to form a cycle. For any edge that has exactly one endpoint in $S$, orient the edge toward $S$. An example is given in Fig.~\ref{fig:orient}. Orient all other edges arbitrarily. For each $i=3,\ldots,n$, this orientation induces a function $g_i:Q_{n-2}\to\{0,1\}$, where $g_i(a_3,\ldots,a_n)$ is determined by the direction of the edge $\{(a_3,\ldots a_{i-1},0,a_{i+1},\ldots,n),(a_3,\ldots a_{i-1},1,a_{i+1},\ldots,n)\}$.
	The nodes' self-independent and historyless reaction functions are as follows.
	
	\begin{align*}
	f_1(a_1,\ldots,a_n)&=
	\left\{
	\begin{array}{ll}
	0&\quad\mbox{if }a_2=1\text{ and }(a_3,\ldots,a_n)=v^j\mbox{ for some }j\in A\\
	1&\quad\mbox{otherwise}
	\end{array}
	\right.\\
	f_2(a_1,\ldots,a_n)&=
	\left\{
	\begin{array}{ll}
	0&\quad\mbox{if }a_1=1\text{ and }(a_3,\ldots,a_n)=v^j\mbox{ for some }j\in B\\
	1&\quad\mbox{otherwise}
	\end{array}
	\right.\\
	f_i(a_1,\ldots,a_n)&=
	\left\{
	\begin{array}{ll}
	g_i(a_3,\ldots,a_n)&\quad\mbox{  if }a_1=a_2=0\\
	1&\quad\mbox{ otherwise}
	\end{array}
	\right.
	\end{align*}
	Informally, the aim is for nodes $3,\ldots,n$ to follow the snake $S$ to vertices corresponding to each value $j\in[q]$.
	
	\begin{figure}[h]
		\centering
		\begin{tikzpicture}[scale=0.8, main node/.style={fill=white,ellipse,draw,thick,minimum height=0.7cm,minimum width=1.2cm,inner sep=0pt,opacity=0.75,text opacity=1}]
		\node[main node] (0001) {0001};
		\node[main node] (0011) [right=1cm of 0001] {0011};
		\node[main node] (0101) [above=1cm of 0001] {0101};
		\node[main node] (0111) [right=1cm of 0101] {0111};
		\node[main node] (1001) [above right=0.25cm and 0.75cm of 0001] {1001};
		\node[main node] (1011) [right=1cm of 1001] {1011};
		\node[main node] (1101) [above=1cm of 1001] {1101};
		\node[main node] (1111) [right=1cm of 1101] {1111};
		\node[main node] (0000) [below left=1cm and 0.5cm of 0001]{0000};
		\node[main node] (0010) [right=3cm of 0000] {0010};
		\node[main node] (0100) [above=4cm of 0000] {0100};
		\node[main node] (0110) [right=3cm of 0100] {0110};
		\node[main node] (1000) [above right=0.25cm and 1.5cm of 0000] {1000};
		\node[main node] (1010) [right=3cm of 1000] {1010};
		\node[main node] (1100) [above=4cm of 1000] {1100};
		\node[main node] (1110) [right=3cm of 1100] {1110};
		\path[draw,ultra thick,->]
		(0000) edge (0001) 
		(0001) edge (0011)
		(0011) edge (1011)
		(1011) edge (1010)
		(1010) edge (1110)
		(1110) edge (0110)
		(0110) edge (0100)
		(0100) edge (0000)
		;
		\path[draw,->,thick,dotted]
		(1111) edge (1101)
		(0010) edge (0000)
		(1000) edge (0000)
		(0010) edge (1010)
		(1111) edge (1110)
		(0111) edge (0110)
		(0101) edge (0100)
		(0111) edge (0011)
		(1111) edge (1011)
		(0010) edge (0011)
		(1100) edge (0100)
		(1100) edge (1110)
		(1000) edge (1010)
		(0101) edge (0001)
		(1001) edge (0001)
		(1001) edge (1011)
		(1100) edge (1101)
		(0101) edge (1101)
		(1001) edge (1101)
		(1000) edge (1001)
		(0101) edge (0111)
		(0111) edge (1111)
		(0010) edge (0110)
		;
		\begin{pgfonlayer}{bg}
		\path[draw,->,thick,dotted]
		(1000) edge (1100)
		;
		\end{pgfonlayer}
		\end{tikzpicture}
		\caption{An acceptable orientation of the edges on $Q_4$. The solid edges form a cycle on a maximal snake $S$, and no edge is directed away from $S$.}
		\label{fig:orient}
	\end{figure}
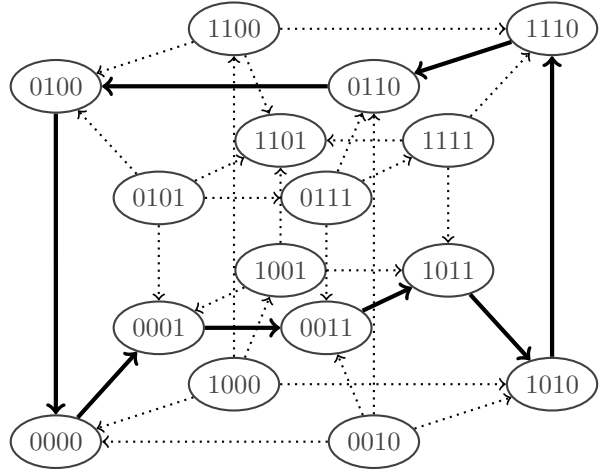
	
	\begin{obs}
		$1^n$ is the unique limit state of the system.
	\end{obs}
	
	In our reduction Alice simulates node 1 (whose reaction function is based on $A$), Bob simulates node 2 (whose reaction function is based on $B$), and one of the two parties simulates all other nodes (whose reaction functions are based on neither $A$ nor $B$). The theorem now follows from the combination of the following two claims.
	
	\begin{claim}\label{claim:infinite}
		In an oscillation there must be infinitely many time steps in which both node 1 and 2's actions are $0$.
	\end{claim}
	\begin{proof}
		Suppose that from some moment forth it is never the case that both node 1 and 2's actions are $0$. Observe that from that time onwards the nodes $3,\ldots,n$ will always choose the action $1$ when activated. Hence, after some time has passed the actions of all nodes in $\{3,\ldots,n\}$ will be $1$. Observe that whenever nodes 1 and 2 are activated thereafter they shall choose the action $1$, so we have convergence to the limit state $1^n$.
	\end{proof}
	
	\begin{claim}
		The system is convergent iff $A\cap B=\emptyset$.
	\end{claim}
	\begin{proof}
		If $A\cap B\neq\emptyset$, initialize the system to a state $(0,0,a_3,\ldots,a_n)$, where $(a_3,\ldots,a_n)=v^1$. Consider a schedule that activates $\{3,\ldots,n\}$ in every timestep until $(a_3,\ldots,a_n)=v^j$ for some $j\in A\cap B$. When that happens, the schedule activates $\{1,2\}$ for two consecutive timesteps, then resumes activating $\{3,\ldots,n\}$. The functions $g_i$ ensure that, for each $j\in[q]$, the vector $(a_3,\ldots,a_n)$ will be equal to $v_j$ within a finite number of timesteps. Since there is some $j\in A\cap B$, nodes 1 and 2 will eventually be activated, so this schedule is fair. This initial state and schedule clearly produce an oscillation, so the system is not convergent in this case.
		
		Now assume for contradiction that  $A\cap B=\emptyset$ and the system is not convergent. we know from Claim~\ref{claim:infinite} that if there is an oscillation, then there are infinitely many time steps in which both node 1 and 2's actions are $0$. We argue that this implies that there must be infinitely many time steps in which both nodes select action $0$ \emph{simultaneously}. Indeed, node 1 only chooses action $0$ if node 2's action is $1$, and vice versa, and so if both nodes never choose $0$ simultaneously, then it is never the case that both nodes' actions are $0$ at the same time step, a contradiction. Now, when is it possible for both $1$ and $2$ to choose $0$ at the same time? Observe that this can only be if the actions of nodes $3,\ldots,n$ constitute an element that is in both $A$ and $B$. Hence, $A\cap B\neq\emptyset$, another contradiction.
	\end{proof}

	We have reduced $\textsc{Disj}^{q}(A,B)$, which requires $\Omega(q)$ bits of communication, to the problem of checking convergence of an $n$-node system with historyless and self-independent reaction functions. As $q$ was defined as the length of a maximal snake in $Q_n$, a classical combinatorial result due to Evdokimov shows that $q=\Omega(2^n)$.
	
	\begin{theorem}[Evdokimov~\cite{Evdokimov1969}]\label{thm:maxsnake}
		Let $z\in\Z_{+}$ be sufficiently large. Then, the size $|S|$ of a maximal snake in the $z$-hypercube $Q_z$ is at least $\lambda2^z$ for some $\lambda>0$.
	\end{theorem}

	This completes the proof of Theorem~\ref{thm:com-cxity}.
	\end{proof}

\subsection{Computational Complexity}

The above communication complexity hardness result required the representation of the reaction functions to (potentially) be exponentially long. What if the reaction functions can be succinctly described? We now present a strong computational complexity hardness result for the case that each reaction function $f_i$ is historyless, and is given explicitly in the form of a Boolean circuit.

\begin{theorem}\label{thm:pspace}
	When the reaction functions are given as Boolean circuits, determining whether a historyless system with $n$ nodes is convergent is PSPACE-complete.
\end{theorem}

\begin{proof}
	Our proof is based on the proof of Fabrikant and Papadimitriou~\cite{FP08} that checking BGP safety is PSPACE-complete. Importantly, that result does not imply Theorem~\ref{thm:pspace}, since the model of asynchronous interaction considered in that work does not allow for simultaneous activation of nodes. We prove our theorem by reduction from the problem of determining whether a linear space-bounded Turing machine (TM) will halt from every starting configuration.
	
	For $n\in Z_+$, let $M$ be a TM that can only access the first $n$ tape cells. Let $Q$ be $M$'s machine state space, $\Gamma$ its tape alphabet, and $\delta:Q\times\Gamma\to Q\times\Gamma\times\{-1,0,1\}$ its transition function. A configuration of $M$ is a triple $(q,\bm{a},j)$ where, $q\in Q$ is a machine state, $\bm{a}\in\Gamma^n$ describes the tape contents, and $j\in[n]$ gives the location of the control head.
	
	Then we say that $\langle M,1^n\rangle$ is in \textsc{shc} (for \underline{s}pace-bounded \underline{h}alting from all \underline{c}onfigurations) if, for every configuration, $M$ will halt if it begins its computation from that configuration. As Fabrikant and Papadimitriou~\cite{FP08} argue, the problem of checking whether a space-bounded TM will accept a blank input is reducible to \textsc{shc}, and thus \textsc{shc} is PSPACE-hard.
	
	We now reduce \textsc{shc} to the problem of determining whether a historyless interaction system is convergent. Given $n\in\Z_+$ and a TM $M$ that can only access the first $n$ tape cells, we construct a historyless system that is convergent if and only if $\langle M,1^n\rangle$ is in \textsc{shc}. The system has a \emph{cell node} for each of the first $n$ tape cells and a \emph{head node} to represent the control head. The cell nodes $1,\ldots,n$ each have the action space $\Gamma$, and the head node $n+1$ has action space $Q\times\Gamma\times[n]\times\{-1,0,1\}$. For $\bm{a}\in\Gamma^n$, each cell node $i\in[n]$ has the reaction function
	$$
	f_i(\bm{a},(q,\gamma,j,d))=
	\left\{
	\begin{array}{ll}
	\gamma&\;\mbox{if }i=j\\
	a_i&\;\mbox{otherwise}\,.
	\end{array}
	\right.
	$$
	For the head node, $f_{n+1}(\bm{a},(q,\gamma,j,d))$ is given by the following procedure.
	
	\begin{algorithm}[H]
		\If{$a_j=\gamma$}{\If{$j+d\in[n]$}{$j\gets j+d$\;}$(q,\gamma,d)\gets \delta(q,a_j)$\;}
		\Return{$(q,\gamma,j,d)$\;}
	\end{algorithm}
	
	Observe that $(\bm{a},(q,\gamma,j,d))$ is a limit state of this system if and only if $q$ is a halting machine state for $M$. Suppose the system is convergent, and let $(q,\bm{a},j)$ be a configuration of $M$. Consider the system trajectory that begins at state $(\bm{a},(q,a_j,j,0))$ and follows the schedule that activates every node in every round. This trajectory will reach a limit state, and it corresponds directly to a halting run of $M$ that begins from $(q,\bm{a},j)$, so $\langle M,1^n \rangle$ is in \textsc{shc}.
	
	Conversely, suppose that $\langle M,1^n \rangle$ is in \textsc{shc}, and let $(\bm{a},(q,\gamma,j,d))$ be a system state. Suppose that $a_j=\gamma$. Then consider the computation by $M$ that begins at configuration $(q,\bm{a},j^\prime)$, where $j^\prime=\min\{\max\{j+d,1\},n\}$. This computation will reach a halting machine state, and any fair trajectory from $(\bm{a},(q,\gamma,j,d))$ will go through the same sequence of machine states as this computation, which means that the trajectory will converge to a limit state. If $a_j\neq\gamma$, then the system state will remain the same until $j$ is activated and takes action $\gamma$, after which the above argument applies. Thus the system is convergent. We conclude that checking whether the system is convergent is PSPACE-hard. Since it is straightforward to check convergence in polynomial space, this problem is PSPACE-complete.
\end{proof}

In a preliminary version of this work~\cite{jsw11ics}, we conjectured that the above PSPACE-completeness result also holds for the case of self-independent reaction functions.
\begin{conjecture}
	Determining whether a system with $n$ nodes, each with a deterministic self-independent and historyless reaction function, is convergent is PSPACE-complete.
\end{conjecture}
This conjecture has since been proved by Engelberg et al.~\cite{EFSW13}.
	\section{Conclusions and Future Research}\label{sec:conc}

In this paper, we have taken the first steps towards a complete understanding of strategic dynamics in distributed settings. We proved a general non-convergence result and several hardness results within this model. We also discussed some important aspects such as the implications of fairness and randomness, as well as applications to a variety of settings. We believe that we have only scratched the surface in the exploration of the convergence properties of game dynamics in distributed computational environments, and many important questions remain wide open. We now outline several interesting directions for future research.

\vspace{1mm}\noindent{\bf Other limitations, convergence notions, and equilibria.} We have considered particular limitations on reaction functions, modes of convergence, and kinds of equilibria. Understanding the effects of asynchrony on different classes of reaction functions (e.g., uncoupled dynamics, dynamics with outdated information) and for other types of convergence (e.g., of the empirical distributions of play) and equilibria (e.g., mixed Nash equilibria, correlated equilibria) is a broad and challenging direction for future research.

\vspace{1mm}\noindent{\bf Other notions of asynchrony.} We believe that better understanding the role of degrees of fairness, randomness, and other restrictions on schedules from distributed computing literature, in achieving convergence to equilibrium points is an interesting and important research direction.

\vspace{1mm}\noindent{\bf Characterizing asynchronous convergence.} We still lack characterizations of asynchronous convergence even for simple dynamics (e.g., deterministic and historyless). Our PSPACE-completeness result in Section~\ref{ssec:complexity} eliminates the possibility of short witnesses of guaranteed asynchronous convergence unless PSPACE $\subseteq$ NP, but elegant characterizations are still possible.

\vspace{1mm}\noindent{\bf Topological and knowledge-based approaches.} Topological~\cite{BG93,HS99,SZ00} and knowledge-based~\cite{HM90} approaches have been very successful in addressing fundamental questions in distributed computing. Can these approaches shed new light on the implications of asynchrony for strategic dynamics?

\vspace{1mm}\noindent{\bf Further exploring the environments of Section~\ref{ssec:examples}.} We have applied our main non-convergence result to the environments described in Section~\ref{ssec:examples}. These environments are of independent interest and are indeed the subject of extensive research. Hence, the further exploration of dynamics in these settings is important.

	\subsection*{Acknowledgments}
	We thank Danny Dolev, Alex Fabrikant, Idit Keidar, Jonathan Laserson, Nati Linial, Yishay Mansour, Yoram Moses, and two anonymous reviewers for helpful discussions and comments. This work was initiated partly as result of the DIMACS Special Focus on Communication Security and Information Privacy. This work was partially supported by NSF grant CCF-1101690, ISF grant 420/12, the Israeli Center for Research Excellence in Algorithms (I-CORE), and the Office of Naval Research.

	\bibliographystyle{plain}
	\bibliography{dcah}
\end{document}